\documentclass[a4paper,onecolumn, 12pt]{IEEEtran}
\usepackage{amsmath,amssymb}
\usepackage{geometry}
\usepackage{graphicx}
\usepackage[justification=centering, font = small, labelsep=space]{caption}
\usepackage{mathrsfs}
\usepackage{longtable}
\usepackage{mathtools}
\usepackage{epstopdf}
\usepackage{dsfont}
\usepackage{amsthm}
\usepackage[hidelinks]{hyperref}
\usepackage{cite}
\usepackage{color}
\usepackage{soul}
\usepackage{setspace}
\usepackage{subcaption}
\usepackage{multicol}
\usepackage[ruled,vlined]{algorithm2e}
\usepackage{lettrine}
\usepackage{algorithmic}
\usepackage{nicefrac,xfrac}
\usepackage{dblfloatfix}

\makeatletter
\def\th@plain{%
\thm@notefont{}% same as heading font
  \itshape % body font
  }
\def\th@definition{%
\thm@notefont{}% same as heading font
  \normalfont % body font
  }
\makeatother
\newtheorem{theorem}{Theorem}

\newtheorem{corollary}{Corollary}
\newtheorem{lemma}{Lemma}
\begin{document} 

\title{ \ \\ Energy Efficiency Maximization in mmWave Wireless Networks with 3D beamforming}
\author{ {\small{ \ \\Mahdi Baianifar \\  }
 }}
\maketitle

\begin{abstract}
In this paper, we address the problem of three dimensional beamforming (3DBF) in millimeter wave (mmWave) wireless networks.~In particular, we study the impact of base station (BS) antenna tilt angle optimization on the energy efficiency (EE) of mmWave networks under two different scenarios: a homogeneous network consisting of multiple macro base stations (MBSs), and a heterogeneous network where several femto base stations are added within the coverage areas of the MBSs.~First, by adopting a stochastic geometry approach, we analyze the coverage probability of both scenarios that incorporate the 3DBF.~Then, we derive the EE of the networks as a function of the MBS antenna tilt angle.~Next, optimization problems are formulated to maximize the EE of the networks by optimizing the tilt angle.~Since the computational complexity of the optimal solution is very high, near-optimal low-complexity methods are proposed for solving the optimization problems.~Simulation results show that in the mmWave networks, the 3DBF technique with optimized tilt angle can considerably improve the EE of the network.~Also, the proposed low complexity approach presents a performance close to the optimal solution but with a significant reduced complexity. 
\end{abstract}


\section{\uppercase{Introduction}}
\label{sec:introd}
\lettrine[lines=2]{I}{ncreasing} demands for high data rate in the 5th generation (5G) cellular systems need much more bandwidth compared to current cellular networks. The millimeter wave (mmWave) frequency bands have recently attracted a lot of attentions due to large bandwidth that they offer \cite{itwillwork13, covratemmwave15}. However, in practice they encounter some challenges including high path loss, high power consumptions and the blockage effect caused by buildings and human bodies \cite{ blockanalysis14,mmWaveEn14,JCNmmWave16}. Another emerging technique in 5G wireless networks is three dimensional beamforming (3DBF) which utilizes active large antenna arrays to control the antenna patterns in a 3D space \cite{DrRazaviTr}. In fact, in the 3DBF more degrees of freedom are exploited to adjust the beam patterns in both horizontal and vertical (tilt angle) domains to improve the network performance in term of spectral efficiency and energy efficiency (EE) \cite{DrRazaviTr, 3DBFdesignLee13}. On the other hand, due to the short wavelength of the mmWave bands, a large number of antenna elements can be packed in a small area arrays which makes them suitable for employing the 3DBF.

One of the recent powerful mathematical techniques that has been proposed for analyzing the performance of cellular networks is stochastic geometry (SG) \cite{covratemmwave15,SGmodelingElsawy17,JCNLTEA18}. This technique is widely used in evaluating different network performance metrics including coverage, capacity, spectral efficiency and the EE in the microwave as well as mmWave systems.
An SG-based mathematical framework to model random blockage in the mmWave networks has been proposed in \cite{blockanalysis14} in which the authors proved that the distribution of the number of the blockages in a link follows a Poisson distribution. In \cite{covratemmwave15}, the SG approach was employed for analyzing the coverage and rate of the mmWave networks and it was shown that the mmWave networks achieve a comparable coverage but higher data rates than microwave networks. Also, the SG technique was employed in \cite{SGmodmultitier15} to evaluate the performance of multi-tier networks and it was shown that a sufficiently dense mmWave cellular network can outperform microwave cellular networks in terms of the coverage probability. In addition, the downlink of a multi-tier heterogeneous mmWave cellular network in a Nakagami fading channel was investigated by a SG approach in \cite{CovHetnetDown17}. In \cite{Onireti2018}, the effect of user association and power control on the coverage and EE of the mmWave system is investigated. The maximization of the EE by considering a constraint on the coverage probability is studied in \cite{Cen2017} which provides insights for deployment of an energy efficient mmWave network. 

In this paper, we address the problem of the 3DBF in the mmWave networks. In particular, our work focuses on the EE maximization in the mmWave networks by tilt angle optimization at the BSs that are equipped with active antenna systems. To the best of our knowledge, this problem has not been studied before in literature. Furthermore, for our analysis, we use a stochastic geometry approach. In this approach, the location of BSs are modelled by a homogeneous Poisson point process (PPP). In addition, we use a modified model for the propagation channel that properly incorporates the existence of blockage effect in the environment. Using the above assumptions and modeling, we first evaluate the signal-to-noise-plus-interference ratio (SINR) coverage probability and then derive the EE of the network as a function of the BSs' tilt angle.

We solve the above problem for two different scenarios. In the first scenario, a homogeneous network is studied where multiple macro base stations (MBSs) serve a number of macro-users. Applying the SG technique, we compute the coverage probability and the EE of the network. Afterwards, the optimum tilt angle that maximizes the EE is found through an optimization problem. Because of the complex form of the objective function, this optimization problem is hard and can not be solved efficiently. The optimal value is then obtained by exhaustive search over the available range of tilt angles. In order to reduce the complexity, we propose an efficient algorithm based on bisection method which has a close performance to the exhaustive search but with a considerably reduced complexity.

The second scenario that we examine in this paper includes a two tier heterogeneous network (HetNet) composed of multiple MBSs and femto base stations (FBSs) which are modeled by two PPPs with different densities. To limit interference, we define a sleep region around each MBS so that the FBSs in the sleep regions do not transmit any signal. Using this idea, the coverage of the network is evaluated and the EE is calculated. Then, the MBS tilt angle and the radius of the sleep regions are jointly optimized through an optimization problem for maximizing the EE. We also propose an efficient method which considerably reduces the computational complexity. It is shown that the proposed efficient method has only a small degradation in the performance with respect to the optimal solution obtained by exhaustive search. In addition, in the second scenario, we provide a lower bound on the coverage probability of the femto users that is very tight.

Finally, through numerical simulations, we evaluate the performance of the proposed schemes and confirm that in a mmWave network, using the 3DBF technique with optimized tilt angle considerably improves the performance of the network in terms of the EE. Our simulations also demonstrate the effectiveness of the proposed low-complexity optimization methods.

The rest of the paper is organized as follows: In Sec. \ref{sec2}, the system models of the homogeneous network and HetNet are described.~Sec. \ref{seccov} derives the coverage probabilities and EE of two scenarios.~In Sec. \ref{EEsec}, the EE maximization problem is formulated and the low-complexity solving method is presented.~Numerical results are presented in Sec. \ref{simsec}, and finally Sec. \ref{consec} concludes the paper.

\section{SYSTEM MODEL}\label{sec2}
We consider downlink of a multi-cell mmWave cellular network under two scenarios: a homogeneous network composed of multiple MBSs, and a two tier HetNet consisting of multiple MBSs and multiple FBSs that both the MBSs and FBSs utilize same frequencies in the mmWave bands. The path loss of the channels between the MBSs and macro users are given by \cite{blockanalysis14}
\begin{equation} \label{PathLossEq}
L \left(r \right) = \left\{
\begin{array}{lcl}
C_L r^{-\alpha_L} & \text{with prob.} & P_L\left(r\right) \\
C_N r^{-\alpha_N} & \text{with prob.} & P_N\left(r \right) 
\end{array}, \right.
\end{equation}
where $C_{\small{L}}$ and $C_{\small{N}}$ account for the path loss in a reference distance for line of sight (LOS) and non-LOS (NLOS) links, respectively, $r$ is the distance between a BS and its associated user, and $\alpha_{\small{L}}$ and $\alpha_{\small{N}}$ denote the path loss exponents for LOS and NLOS links, respectively. Links are in the LOS condition with probability $P_L\left(r\right) = e^{-\beta r}$ where $\beta$ indicates the intensity of the blockage effect; and in the NLOS condition with probability $P_N\left(r\right)=1-P_L\left(r\right)$ \cite{blockanalysis14}. Also, it is assumed that the link between the MBSs and the femto users and that between FBSs and macro users are always in the NLOS condition. In addition, since femto users are usually located indoor, the channel between a femto user and interfering (non-serving) FBSs are assumed to be NLOS. 

To design the 3DBF techniques, we need a model for the vertical and horizontal antenna patterns at the MBSs. In this paper, for the vertical plane, we adopt a model presented in \cite{3GPPstd} in which each MBS's antenna gain is expressed as 
\begin{equation}
G \left( \theta , \theta_{\text{tilt}} \right) = - \min{\left( 12\left( \frac{\theta - \theta_{\text{tilt}}}{\theta_{\text{3dB}}}\right)^2, \text{{SLL}}_\text{dB}\right)} \ {\text{dB}}, \label{AntennaPattMBSs}
\end{equation}
where $\theta \ge 0$ is the angle between the horizon and the line connecting the MBS to the user (see Fig.~\ref{FigVertPat}). In addition, $\theta_{\text{tilt}}\ge 0$, $\theta_{\text{3dB}}$, and $\text{{SLL}}_{\text{dB}}$ stand for the array tilt angle, the $\text{3dB}$ beamwidth, and the side-lobe level of the MBS antenna pattern in the vertical plane, respectively \cite{tiltangtwo15}. By defining $H_{\text{eff}} \triangleq H_{\text{BS}} - H_{u}$; where $H_{\text{BS}}$ and $H_u$ represent the MBSs' and users' antenna heights, respectively, \eqref{AntennaPattMBSs} can be rewritten as
\begin{equation*}
G \left(R,\theta_{\text{tilt}} \right) = -\min{\left(12\left(\frac{{\text{atan}}\left({H_{\text{eff}}}/{R}\right)-\theta_{\text{tilt}}}{\theta_{\text{3dB}}}\right)^2, \text{{SLL}}_{\text{dB}}\right)} \text{dB},
\end{equation*}
\noindent where $R$ equals the horizontal distance between the MBS and the user. It is assumed that all FBSs’ and the users’ antennas have an omni-directional pattern in the vertical domain. 
\begin{figure}[t!]
\centering 
\includegraphics[scale=0.45]{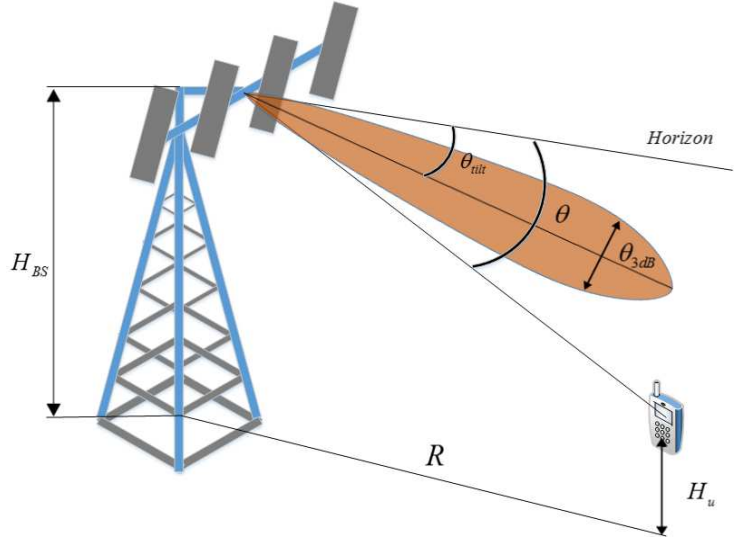}
\caption{~Vertical antenna pattern at each BS.} \label{FigVertPat}
\end{figure}

For modeling the MBS's and macro users' antennas horizontal pattern, a sectorized pattern is utilized that has constant gains of $M$ and $m$ in its main-lobe and side-lobe, respectively \cite{covratemmwave15}. The total antenna gain of a transmitter to receiver link in the horizontal plane is modelled by a random variable $D$ which takes four values of $d_1 = M_t M_r$, $d_2 = M_t m_r$, $d_3 = m_t M_r$, and $d_4 = m_t m_r$ with probabilities $p_1 = c_t c_r$, $p_2 = c_t \left( 1 - c_r \right)$, $p_3 = \left( 1 - c_t \right) c_r$, and $p_4 = \left( 1-c_t \right) \left(1-c_r \right)$, respectively. The subscripts $t$ and $r$ stand for the transmitter (MBS) and receiver (macro user), respectively. In addition, we have $c_t = {\varphi_t}/{2\pi}$ and $c_r = {\varphi_r}/{2\pi}$, in which $\varphi_t$ and $\varphi_r$ indicate the horizontal beamwidth of the transmitter and the receiver antennas, respectively. Also, we assume that antennas horizontal pattern for FBS and femto users are $M^f$ and $m^f$ in its main-lobe and side-lobe, respectively with $\phi_t^f$ and $\phi_r^f$ as a horizontal beamwidth of the FBSs and femto users, respectively. We denote total antenna gain of the FBS and femto user by $D^f$ which takes value $d_i^f$ with probabilities $p_i^f$ for $i=1,...,4$, where they can be calculated in similar manner as the MBS and macro user. 
In the following, we explain these two scenarios for the mmWave network.

\subsection{Homogeneous Cellular Network}
In this scenario, only MBSs exist in the network whose positions are modeled by a homogeneous PPP $\Phi_m$ with density $\lambda_m$. From the Slivnyak theorem \cite{SGwireles09}, to evaluate the performance of the network, it is sufficient to consider a typical user located at the origin and analyze its performance. The received signal at the typical user can be written as

\begin{align} \label{REcForMBSonly}
y \ & = \ \sqrt{P_m L_m\left(r_{0,0}\right) D_0 G_0} \ h_{0,0} \ s_{0} \qquad \qquad \qquad \qquad \nonumber \\
&+ \sum_{j \neq 0,\: X_j \in \Phi_m}{\sqrt{P_m L_m\left(r_{j,0}\right) D_j G_j} \ h_{j,0} \ s_{j}}+n,
\end{align} 
where $P_m$ represents the transmission power of each MBS and $r_{j,0}$, $L_m\left(r_{j,0} \right)$ and $h_{j,0}$ indicate the distance, the path loss and the small scale fading between the $j$th MBS ($j=0$ is for the MBS that serves the typical user) and the typical user, respectively, $D_j$ is the total horizontal antenna gain between the $j$th MBS and the typical user, $G_j = G \left(r_{j,0},\theta_{\text{tilt}} \right)$ shows the vertical antenna gains, $s_{j}$ equals the transmitted signal of the $j$th MBS, $n\sim {\mathcal{CN}} \left(0 ,\sigma^2 \right)$ stands for the noise, and the location of the $j$th MBS is denoted by $X_j$. We consider Nakagami-m fading in which fading power $\left\lvert h_{j,0} \right\rvert^2$ follows a Gamma distribution $\Gamma\left(m,1/m \right)$ with the following complementary cumulative distribution function (CCDF)
\begin{equation*}
\bar{F}\left(z \right) = e^{-mz} \sum_{k=0}^{m-1}{\frac{\left(mz\right)^k}{k! }}.
\end{equation*}
In all equations, index $0$ is used for identifying the typical user and also the MBS that serves this user.

\subsection{Two Tier Heterogeneous Network}
In this scenario, in addition to the MBSs, a number of FBSs exist in the network. The locations of the MBSs and FBSs are modeled by two independent homogeneous PPP $\Phi_m$ and $\Phi_f$ with densities $\lambda_m$ and $\lambda_f$, respectively. Location of the $j$th FBS is denoted by $Y_j$. In this case, we analyze the performance of a typical macro user as well as a typical femto user. It is assumed that femto users are uniformly distributed within the coverage area of its serving FBS, which has a circular area of radius $R_f$. Defining the signal attenuation caused by walls as $\ell_W$, the attenuation of the links between the MBSs and the typical indoor femto user and that between non-serving FBSs and the typical indoor femto user are represented as $\ell_W$ and ${\left(\ell_W \right)}^2 $, respectively. To decrease interference on the macro users, we consider a sleep region with radius $R_c$ around each MBS, where the FBSs lying in this region are forced to enter a sleep mode and do not transmit any signals. In other words, if the distance of a FBS from each MBS is less than $R_c$, it will be turned off. By considering the above assumptions, the received signal at the typical macro user and femto user respectively become
\begin{align} \label{TwoTierMBS}
y_m = & \sqrt{P_m L_m\left( r^m_{0,0} \right) D_0 G_0} h_{0,0} s_0 \nonumber \\
& +\sum_{j \neq 0,\: X_j \in \Phi_m}{\sqrt{P_m L_m\left( r^m_{j,0} \right) D_j G_j} h_{j,0} s_j} \nonumber \\
&+\sum_{j,\: Y_j \in \Phi'_f }{\sqrt{P_f \ell_W L_m^f\left( r^{fm}_{j,0} \right) D^{fm}_j} g_{j,0}^f x_j}+n,
\end{align}
\begin{align} \label{TwoTierFemto}
y_f = & \sqrt{P_f L_f\left( r^f_{0,0} \right) D_0^f} g_{0,0} x_0 \nonumber \\
& +\sum_{j,\: X_j \in \Phi_m }{\sqrt{P_m \ell_W L_f^m\left( r^{mf}_{j,0} \right) D_j^{mf} G_j} h^m_{j,0} s_j} \nonumber \\
&+\sum_{j\neq 0,\: Y_j \in \Phi_f^{'} }{\sqrt{P_f \left(\ell_W \right)^2 L_f\left( r^f_{j,0} \right) D_j^f} g_{j,0} x_j}+n',
\end{align}
where $r^m_{j,0}$ and $r^{fm}_{j,0}$ denote the distance between the $j$th MBS and the typical macro user and that between the $j$th FBS and the typical macro user, respectively, $r^f_{j,0}$ and $r^{mf}_{j,0}$ show the distance between the $j$th FBS and the typical femto user and that between the $j$th MBS and the typical femto user, respectively, $\Phi'_f$ indicates the modified version of $\Phi_f$ after excluding the FBSs in radius of $R_c$ of each MBS, $L_m^f \left( r_{j,0} \right)$, $D_j^{fm}$ and $g^f_{j,0}$ represent the path loss, the total antenna gain in the horizontal domain and the small scale fading between the typical macro user and the $j$th FBS, respectively, $x_j$ stands for the transmitted signal of the $j$th FBS, $g_{j,0}$ and $L_f\left(r_{j,0}^f\right)$ equal the small scale fading and the path loss between the typical femto user and the $j$th FBS, respectively, $D_j^f$ shows the total antenna gain between the $j$th FBS and the typical femto user, $L_f^m \left( r_{j,0}^{mf} \right)$, $D_j^{mf}$ and $h^m_{j,0}$ represent the path loss, the total antenna gain in the horizontal domain and the small scale fading between the typical femto user and the $j$th MBS, respectively, $P_f$ equals the transmitted power of the $j$th FBS, and $n$ and $n'$ are the complex Gaussian noise as ${\mathcal{CN}}\left(0,\sigma^2 \right)$. We define the values of variable $D^{fm}$ by $d_i^{fm}, i=1,..., 4$ with probabilities $p_i^{fm}$ (calculations are similar to MBS and macro users total antenna gain). $D^{mf}$ takes values $d_i^{mf}$ with probabilities $p_i^{mf}$. Also, since these links are in NLOS condition, $g_{j,0}, g^f_{j,0}, h^m_{j,0}$ are distributed as ${\mathcal{CN}} \left(0,1 \right)$. 

\section{ENERGY EFFICIENCY CALCULATION}\label{seccov}
In this section, first the coverage probability of the network is calculated and then used for deriving the EE under two above scenarios. There exist different user association rules like the nearest BS, minimum path or the strongest average power and also maximum SINR \cite{rtackappandrew11, covratemmwave15, Ktierhetnet12} and in this paper, we use the maximum average received power user association rule in which each user is associated with the BS that provides it the strongest average received power. However, it should be noted that in the mmWave networks, because of the blockage effect and different path loss exponents for LOS and NLOS conditions, the strongest BS is not necessarily the nearest BS. To address this issue, in our analyses, we first map each NLOS BS located at distance $r$ from the origin to an equivalent LOS BS with a larger distance $R_{eq}\left(r\right) = \left(\nicefrac{C_L}{C_N} \right)^{{1}/{\alpha_L}} r^{\sfrac{\alpha_N}{\alpha_L}}$, and then use the distance criterion to associate the users to the BSs. In addition, by considering different probabilities for LOS and NLOS links as in \eqref{PathLossEq}, the homogeneous PPP, $\Phi_m$ can be divided into two independent non-homogeneous PPPs, $\Phi_L$ and $\Phi_N$ with densities $\lambda_L \left( r \right) = \lambda_m P_L \left(r \right) $ and $\lambda_N \left( r \right) = \lambda_m P_N \left(r \right) $, respectively. The distance between the typical user and its serving BS is a random variable whose probability density function (PDF) is obtained by the following lemma.
\begin{lemma} 
Assuming the highest power user association rule, the PDF of the distance between the typical user and its serving BS is given by \eqref{PDFNBBL} at the top of the next page, where $R^{-1}_{eq}\left(r\right) = \mu r^{\kappa}$, $\mu = \left( \nicefrac{C_N}{C_L}\right)^{\sfrac{1}{\alpha_N}}$ and $\kappa = \nicefrac{\alpha_L}{\alpha_N}$.

\begin{figure*}[!t] 
\begin{align} \label{PDFNBBL}
f_R \left(r \right) = & 2\pi \lambda_m \left( r e^{-\beta r} + \mu \kappa r^{2 \kappa -1} \left(1-e^{-\beta R_{eq}^{-1}\left(r \right)} \right) \right) \times \exp \left( -\frac{2\pi \lambda_m}{\beta^2} \left(1-\left(1+\beta r \right) e^{-\beta r} \right)\right) \times\nonumber \\
& 
\exp \left( -\frac{2\pi \lambda_m}{\beta^2} \left( \frac{\beta^2 \left( R^{-1}_{eq}\left(r\right) \right)^2 }{2}+\left(\beta R^{-1}_{eq}\left(r\right) +1 \right)e^{-\beta R^{-1}_{eq}\left(r\right)}-1 \right) \right)
\end{align}
\hrulefill
\end{figure*}

\end{lemma}
\begin{proof}
The CCDF of the distance of the nearest BS to the typical user $R$ can be calculated as
\begin{align}
&\text{Pr} \left\lbrace R>r \right\rbrace = \text{Pr} \left\lbrace \Phi_m \left( B \left(0,r \right) \right) = 0 \right\rbrace \nonumber \\
& = \text{Pr}\left\lbrace \left( \Phi_L \left( B \left(0,r \right) \right) = 0 \right) \cap \left( \Phi_{N} \left(B\left(0,R^{-1}_{eq} \left(r \right) \right) \right) = 0 \right) \right\rbrace \nonumber \\
& = \text{Pr}\left\lbrace \Phi_L \left( B \left(0,r \right) \right) = 0 \right\rbrace \text{Pr}\left\lbrace\Phi_{N} \left(B\left(0,R^{-1}_{eq} \left(r \right) \right) \right) = 0\right\rbrace, \label{CCDFEeq}
\end{align}
where $B\left(0,r \right)$ shows a ball centering at origin with radius $r$, $\Phi_m\left( B\left( 0,r \right)\right)$ represents the number of PPP $\Phi_m$ in the ball $ B\left( 0,r \right)$, and the last equality comes from the fact that $\Phi_L$ and $\Phi_{N}$ are independent. Hence
$\text{Pr}\left\lbrace \Phi_L \left( B \left(0,R \right) \right) = 0 \right\rbrace$ can be calculated as
\begin{align}
&\text{Pr}\left\lbrace \Phi_L \left( B \left(0,r \right) \right) = 0 \right\rbrace \stackrel{\left(a \right)}{=} \exp\left(- \int_{B\left(0,r\right)}{\lambda_L\left( \left\| x\right\| \right) dx} \right) \nonumber \\
&\stackrel{\left(b \right)}{=} \exp \left(-2\pi \lambda_m \int_0^r{\rho P_L \left(\rho \right) d\rho} \right) \nonumber \\
=& \exp \left( -\frac{2\pi \lambda_m}{\beta^2} \left(1-\left(1+\beta r \right) e^{-\beta r} \right)\right), \label{CCDFeqL}
\end{align}
where $\left(a\right)$ is due to the null probability \cite{HaenggiLarge09} and $\left(b\right)$ comes from the definition of $\lambda_L$ and $P_L\left(r \right)$. Following a similar approach, we can calculate $\text{Pr}\left\lbrace \Phi_N \left( B \left(0,R^{-1}_{eq}\left(r\right) \right) \right) = 0 \right\rbrace$.
%\begin{align}
%&\text{Pr}\left\lbrace \Phi_N \left( B \left(0,R^{-1}_{eq}\left(r\right) \right) \right) = 0 \right\rbrace = e^{C} \times \nonumber \\
%&\exp \left( - C \left( \frac{\beta^2 \left(R^{-1}_{eq}\left(r\right) \right)^2 }{2}+\left(\beta R^{-1}_{eq}\left(r\right) +1 \right)e^{-\beta R^{-1}_{eq}\left(r\right)} \right) \right), \label{CCDFeqN}
%\end{align}
%where $C = \frac{2\pi \lambda_m}{\beta^2}$. 
Finally, by inserting into (\ref{CCDFEeq}) and considering the fact that $f_R \left( r \right) = - \frac{d}{dr}\text{Pr}\left\lbrace R>r \right\rbrace$, the proof is completed.
\end{proof}
\subsection{Homogeneous Cellular Network}
\noindent From \eqref{REcForMBSonly}, the received SINR at the typical user is obtained as
\begin{equation*}
\text{SINR} = \frac{P_m L_m \left( r_{0,0} \right) D_0 G_0 \lvert h_{0,0} \rvert^2}{\sum_{j\neq 0,\: X_j \in \Phi_m}{P_m L_m \left( r_{j,0} \right) D_j G_j \lvert h_{j,0} \rvert^2}+\sigma^2} \ .
\end{equation*}
It is assumed that the main beam of the typical user and its serving MBS's antennas are aligned, and therefore $D_0 = M_t M_r$. Then, the coverage probability is calculated in the following theorem.
\begin{theorem}\label{THeo1}
The coverage probability of the typical user associated with the MBS that provides the highest received power is obtained as
\begin{align}\label{PcMBSmaineq}
&{\mathcal{P}}^c \left(\gamma , \theta_{\text{tilt}}\right) = \text{Pr}\left\lbrace \text{SINR}>\gamma \right\rbrace = \qquad \qquad \qquad \qquad \nonumber \\ 
& \int_0^{\infty}{e^{-m s \sigma^2} \sum_{k=0}^{m-1}{\sum_{\ell = 0}^k{C_{k,\ell}} \left[\frac{d^{\ell}}{d z^{\ell}} {\mathcal{L}}_{I_{\Phi_m}}\left(z,\theta_{\text{tilt}} \right) \right]_{z = m s}} f_R \left(\rho \right) d\rho},
\end{align}
where $s$ is defined as $s =\frac{\gamma \rho^{\alpha_L}}{ P_m C_L D_0 G_0 }$ and $C_{k,\ell} = \left(-1 \right)^{\ell} \frac{\left(m s \sigma^2 \right)^{k-\ell}}{k!} {k \choose \ell}$, $\gamma$ is the SINR threshold for the typical user, and ${\mathcal{L}}_{I_{\Phi_m}} \left( z,\theta_{\text{tilt}} \right) $ represents the Laplace transform of $ I_{\Phi_m}(\theta_{\text{tilt}}) = \sum_{j\neq 0,\: X_j \in \Phi_m}{P_m L_m\left(r_{i,0}\right)D_j G_j \lvert h_{j,0} \rvert^2 }$ as 
\begin{align} \label{LaplaceofIntm} 
&{\mathcal{L}}_{I_{\Phi_m}}= E_{I_{\Phi_m}}\left[ \exp \left(-z I_{\Phi_m}\left( \theta_{\text{tilt}} \right) \right) \right] = \nonumber\\
& \prod_{i=1}^4 {\exp \left( -C_i \int_{\rho}^{\infty}{ F_L\left(z,x,d_i,\theta_{\text{tilt}} \right) x P_L \left(x \right) dx } \right) } \times \nonumber \\
& \prod_{i=1}^4 {\exp \left( -C_i \int_{R^{-1}_{eq}(\rho)}^\infty{F_N\left(z,x,d_i ,\theta_{\text{tilt}}\right) xP_N \left(x \right) dx}\right)}.
\end{align} 
Here we define $C_i = 2\pi \lambda_m p_i$ and
\begin{equation*}
F_w\left(z,x,d_i ,\theta_{\text{tilt}} \right) =1- \frac{1}{\left(1+\frac{z P_m C_w d_i G\left(x,\theta_{\text{tilt}} \right)}{m x^{\alpha_w}} \right)^m} \ , \ w \in \lbrace L,N \rbrace.
\end{equation*} 
\end{theorem}
\begin{proof}
${\mathcal{P}}^c$ can be obtained as
\begin{align*}
&{\mathcal{P}}^c \left(\gamma , \theta_{\text{tilt}}\right) = E_{\rho,I_{\Phi_m}} \left\{ \text{Pr} \left\{ \text{SINR} > \gamma \, \bigg{\vert} \, r_{0,0} = \rho \right\} \right\} \nonumber \\
& \stackrel{\left( a \right) }{=} E_{\rho,I_{\Phi_m}} \left\{ \text{Pr} \left\{ \lvert h_{0,0} \rvert^2 > \frac{\gamma {\rho}^{\alpha_L} }{ P_m C_L D_0 G_0} \left( I_{\Phi_m} + \sigma^2 \right) \bigg{\vert} r_{0,0} = \rho \right\} \right\} \nonumber \\
&\stackrel{\left( b \right) }{=} E_{\rho,I_{\Phi_m}}\left\lbrace e^{-m s\left(I_{\Phi_m}+\sigma^2 \right)} \sum_{k=0}^{m-1}{\frac{\left(m s \left(I_{\Phi_m}+\sigma^2 \right)\right)^k}{k!}}\right\rbrace \nonumber \\
&= \int_0^{\infty}{e^{-m s \sigma^2} \sum_{k=0}^{m-1}{\sum_{\ell = 0}^k{C_{k,\ell}} \left[\frac{d^{\ell}}{d z^{\ell}} {\mathcal{L}}_{I_{\Phi_m}}\left(z,\theta_{\text{tilt}} \right) \right]_{z = m s}} f_R \left(\rho \right) d\rho},
\end{align*}
where $E \left\lbrace . \right\rbrace$ denotes the expectation operator and $\left( a \right)$ follows from $L_m\left(r_{0,0}\right) = C_L r_{0,0}^{-\alpha_L}$ for the maximum received power association method, and $\left( b \right)$ comes from the fact that $\lvert h_{0,0} \rvert^2 \sim \Gamma\left( m ,\frac{1}{m} \right)$. 

Also, $I_{\Phi_m} = \sum_{j\neq 0,\: X_j \in \Phi_L }{P_m C_L r_{j,0}^{-\alpha_L} D_j G_j \lvert h_{j,0} \rvert^2}+ \sum_{j\neq 0,\: X_j \in \Phi_N}{P_m C_N r_{j,0}^{-\alpha_N} D_j G_j \lvert h_{j,0} \rvert^2} = I_{\Phi_L} + I_{\Phi_N}$.

Since $\Phi_L$ and $\Phi_N$ are independent, ${\mathcal{L}}_{I_{\Phi_m}}\left( s,\theta_{\text{tilt}} \right)$ can be written as
\begin{align} \label{LaplaceTransEq}
{ \mathcal{L}}_{I_{\Phi_m}} \left( s,\theta_{\text{tilt}} \right) &= E_{I_{\Phi_m}} \left\lbrace \exp \left(- s I_{\Phi_m} \right) \right\rbrace \nonumber \\
&= E_{I_{\Phi_L}} \left\lbrace \exp \left(- s I_{\Phi_L} \right) \right\rbrace E_{I_{\Phi_N}} \left\lbrace \exp \left(- s I_{\Phi_N} \right) \right\rbrace \nonumber \\
&= {\mathcal{L}}_{I_{\Phi_L}}\left( s \right) {\mathcal{L}}_{I_{\Phi_N}}\left( s \right).
\end{align}
Hence we calculate ${\mathcal{L}}_{I_{\Phi_L}}\left( s \right)$ as
\begin{align*}
&{\mathcal{L}}_{I_{\Phi_L}} = E_{I_{\Phi_L}} \left\{ \exp \left(-s \sum_{\substack{j\neq 0,\\ X_j \in \Phi_L}}{P_m C_L r_{j,0}^{-\alpha_L} D_j G_j \lvert h_{j,0} \rvert^2} \right) \right\} \nonumber \\
& = E_{\Phi_L, h_{j,0}, D_j} \left\{ \prod_{\substack{j\neq 0,\\ X_j \in \Phi_L}}{\exp{ \left(-s P_m C_L r_{j,0}^{-\alpha_L} D_j G_j \lvert h_{j,0} \rvert^2 \right) }} \right\} \nonumber \\
& \stackrel{\left( a \right) }{=} E_{\Phi_L} \left\{ \prod_{\substack{j\neq 0,\\ X_j \in \Phi_L} }{{E} \left\{ \frac{1}{\left(1+\frac{s}{m} P_m C_L r_{j,0}^{-\alpha_L}D_j G_j \right)^m} \right\}} \right\} \nonumber \\
& \stackrel{\left( b \right) }{=} \prod_{i=1}^4 {\exp \left( -C_i \int_{\rho}^\infty F_L\left(s,v,d_i,\theta_{\text{tilt}} \right) v P_L\left( v \right) dv \right)}, 
\end{align*}
where $\left( a \right)$ and $\left( b \right)$ are derived from the fact $\lvert h_{j,0} \rvert^2 \sim \Gamma \left(m,\frac{1}{m}\right) $ and the definition of the total antenna gain in horizontal domain, and also from the probability generating functional (PGFL) of the PPP \cite{SGwireles09}. Then ${\mathcal{L}}_{I_{\Phi_N}}\left( s \right)$ is computed by a similar method. Substituting ${\mathcal{L}}_{I_{\Phi_L}}\left( s \right)$ and ${\mathcal{L}}_{I_{\Phi_N}}\left( s \right)$ into \eqref{LaplaceTransEq} with $s = \frac{\gamma {\rho}^{\alpha_L}}{P_m C_L D_0 G_0}$, the proof is completed.
\end{proof}
In the following, by using the coverage probability in \eqref{PcMBSmaineq}, the EE of the network is calculated. The EE is defined as \cite{EETotuGen,Largantensys15}
\begin{align}\label{EEMBSHom} 
\text{EE}(\theta_{tilt}) = \frac{{\mathcal{P}}^c\left(\gamma , \theta_{\text{tilt}}\right) \log_2\left( 1+\gamma \right)}{P_{cm} + \eta_m P_m},
\end{align}
\noindent where $P_{cm}$ indicate the power consumption related to the signal processing and cooling, and $\eta_m$ is the power amplifier efficiency of each BS. By substituting \eqref{PcMBSmaineq} and \eqref{LaplaceofIntm} in \eqref{EEMBSHom}, we have

\begin{align}\label{EEMBSHom2}
&\text{EE}(\theta_{tilt}) =\nonumber \\
&\frac{\int_0^\infty {e^{- ms \sigma^2 } E_{I_{\Phi_m}}\left[ e^{-ms I_{\Phi_m}} \sum_{k=0}^{m-1}\frac{\left( ms \left(I_{\Phi_m}+\sigma^2 \right) \right)^k}{k!} \right] f_R \left(\rho \right) d\rho}}{ P_{cm}+\eta_m P_m}.
\end{align}

\subsection{Two Tier Heterogeneous Network}
In this scenario, it is assumed that location of the MBSs and FBSs are modeled by two independent PPP, $\Phi_m$ and $\Phi_f$ with densities $\lambda_m$ and $\lambda_f$, respectively. According to equations \eqref{TwoTierMBS} and \eqref{TwoTierFemto}, the SINR in the typical macro and femto users are given as
\begin{align*}
\text{SINR}_m & = \frac{P_m L_m\left( r^m_{0,0} \right) D_0 G_0 \lvert h_{0,0} \rvert^2}{I_{\Phi_m}^m + I_{\Phi_f^{'}}^m +\sigma^2}, \\
\text{SINR}_f & = \frac{P_f L_f\left( r^f_{0,0} \right) D_0^f \lvert g_{0,0} \rvert^2}{I_{\Phi_m}^f + I_{\Phi_f^{'}}^f +\sigma^2},
\end{align*}
where $I_{\Phi_m}^m = \sum_{\substack{j\neq 0,\\ X_j \in \Phi_m}}{P_m L_m\left( r^m_{j,0} \right) D_j G_j \lvert h_{j,0} \rvert^2}$, $I_{\Phi_f^{'}}^m = \sum_{\substack{j,\: Y_j \in \Phi'_f } }{P_f \ell_W L_m^f\left( r^{fm}_{j,0} \right) D_j^{fm} \lvert g^f_{j,0} \rvert^2}$, $I_{\Phi_m}^f = \sum_{\substack{j,\: X_j \in \Phi_m}}$ ${P_m \ell_W L_f^m\left( r^{mf}_{j,0} \right) D^{mf}_j G_j \lvert h^m_{j,0} \rvert^2}$ and $I_{\Phi_f^{'}}^f = \sum_{\substack{j\neq 0,\\ Y_j \in \Phi'_f}}P_f \left(\ell_W \right)^2$ $ L_f\left( r^f_{j,0} \right) D_j^f \lvert g_{j,0} \rvert^2$. It is assumed that the main beam of the typical femto user and its serving FBS’s antennas are aligned i.e., $D_0^f = M_t^f M_r^f$.

Because of sleep regions, the FBSs are modeled by a Poisson hole process. According to \cite{HaenggiLarge09}, $\Phi'_f$, has the density
\begin{equation*}
\lambda_{f'} = \lambda_f \exp\left( - \lambda_m \pi R_c^2 \right). \label{lambfprim}
\end{equation*}
The following theorem provides the coverage probability of the macro users.
\begin{theorem}\label{THeo2}
In the HetNet scenario, by considering the sleep region around each MBS, the coverage probability of a typical macro user is expressed as
\begin{align}\label{MacromainCov}
&{\mathcal{P}}_m^c \left(\gamma_m , \theta_{\text{tilt}},R_c\right) = \nonumber \\
& \int_0^\infty{\sum_{k=0}^{m-1}{\sum_{\ell = 0}^k{C_{k,\ell} \left[\frac{d^{\ell}}{d z^{\ell}} {\mathcal{L}}_{I_{\Phi_{m,f}}}\left(z,\theta_{\text{tilt}} \right) \right]_{z = ms} }}} f_R \left(\rho \right) d\rho,
\end{align}
where ${\mathcal{L}}_{I_{\Phi_{m,f}}}\left(z,\theta_{\text{tilt}} \right) = \left( {\mathcal{L}}_{I_{\Phi_m}}\left(z,\theta_{\text{tilt}} \right) {\mathcal{L}}_{I^m_{\Phi'_f}}\left(z,\theta_{\text{tilt}} \right) \right)$ and $s = \frac{\gamma_m {\rho}^{\alpha_L}}{P_m C_L D_0 G_0}$, $\gamma_m$ represents the SINR threshold for the typical macro user and ${\mathcal{L}}_{I_{\phi'_f}^m}\left(z,\theta_{\text{tilt}}\right)$ indicates the Laplace transform of the interference from the FBSs to the macro user as
\begin{align}\label{LapIntFBStoMacuser}
&{\mathcal{L}}_{I_{\phi'_f}^m} \nonumber \\
& = \prod_{i=1}^4{\exp \left(-C_i^{f'} \left(s P_f \ell_W C_N d_i^{fm} \right)^{\frac{2}{\alpha_N}} \frac{\pi}{\alpha_N \sin\left(\frac{2\pi}{\alpha_N} \right)} \right)},
\end{align}
where $C_i^{f'} = 2\pi \lambda_{f'} p_i^{fm}$.
\end{theorem}
\begin{proof}
Adopting a similar approach to theorem \ref{THeo1}, we have
\begin{align*}
& {\mathcal{P}}_m^c \left(\gamma_m , \theta_{\text{tilt}},R_c\right) =\text{Pr} \left\{ \text{SINR}_m > \gamma_m \right\} \nonumber \\
&=\int_0^{\infty}{\text{Pr}{ \left\lbrace \frac{P_m C_L r^{-\alpha_L} D_0 G_0 \lvert h_{0,0} \rvert^2}{I_{\Phi_m}^m + I^m_{\Phi'_f} + \sigma^2} >\gamma_m \bigg{\vert} r =\rho \right\rbrace } f_R \left( \rho \right) d\rho}.
\end{align*}
We can compute the probability inside the integral as
\begin{align*}
&\text{Pr}{ \left\lbrace \frac{P_m C_L r^{-\alpha_L} D_0 G_0 \lvert h_{0,0} \rvert^2}{I_{\Phi_m}+I^m_{\Phi'_f} + \sigma^2} >\gamma_m \mid r =\rho \right\rbrace} \nonumber \\
& = \sum_{k=0}^{m-1}{\sum_{\ell = 0}^k{C_{k,\ell} \left[\frac{d^{\ell}}{d z^{\ell}} \left( {\mathcal{L}}_{I_{\Phi_m}}\left(z,\theta_{\text{tilt}} \right) {\mathcal{L}}_{I^m_{\Phi'_f}}\left(z,\theta_{\text{tilt}} \right) \right) \right]_{z = ms} }},
\end{align*}

where we use the fact that $\lvert h_{0,0} \rvert^2$ has Gamma distribution. 

Similarly, we have 
\begin{align*}
&{\mathcal{L}}_{I^m_{\Phi'_f}} \nonumber \\
&= E \left\lbrace \exp \left(-s \sum_{\substack{j,\: Y_j \in \Phi'_f }}{P_f \ell_W C_N \left( r^{fm}_{j,0}\right) ^{-\alpha_N} D_j^{fm} \lvert g^f_{j,0} \rvert^2} \right) \right\rbrace \nonumber \\
&\stackrel{\left(a \right)}{=} E_{\Phi'_f, D_j^{fm}} \left\lbrace \prod_{\substack{j,\: Y_j \in \Phi'_f }}{\frac{1}{1+sP_f \ell_W C_N \left( r^{fm}_{j,0}\right) ^{-\alpha_N} D_j^{fm}}} \right\rbrace \nonumber \\
& \stackrel{\left( b \right)}{=} \prod_{i=1}^4 \exp{\left( -C_i^{f'} \int_0^{\infty}{\frac{x}{1+\left(s P_f \ell_W C_N d_i^{fm} \right)^{-1}x^{\alpha_N}} dx}\right)},
\end{align*}
where $\left( a \right)$ is due to the fact $\lvert g_{j,0}^f \rvert^2 \sim \text{exp} \left(1 \right) $ and $\left( b \right)$ is derived from the PGFL of the PPP $\Phi'_f$ \cite{SGwireles09}.
\end{proof}
Next, we derive the coverage probability of the typical femto user in the following theorem.
\begin{theorem}
In the HetNet scenario, the coverage probability of the typical femto user is expressed as
\begin{align}\label{FemUsermainCovP}
&{\mathcal{P}}_f^c \left(\gamma_f , \theta_{\text{tilt}},R_c\right) = \exp\left(-\pi \lambda_m R_c^2 \right) \times \nonumber \\
&\int_0^{R_f}{e^{-s_f\sigma^2} {\mathcal{L}}_{I_{\Phi_m}^f}\left(s_f,\theta_{\text{tilt}}\right) {\mathcal{L}}_{I^f_{\phi'_f}}\left(s_f, \theta_{\text{tilt}} \right) g_R \left(\rho \right) d\rho d\rho},
\end{align}
where $s_f = \frac{\gamma_f}{P_f C_L D_0^f} {\rho}^{\alpha_L}$, $\gamma_f$ denotes the SINR threshold for the typical femto user, and $ {\mathcal{L}}_{I_{\Phi_m}^f}\left(s_f,\theta_{\text{tilt}} \right)$ and ${\mathcal{L}}_{I^f_{\Phi'_f}}\left(s_f,\theta_{\text{tilt}} \right)$ are obtained as
\begin{align}
&{ \mathcal{L}}_{I_{\Phi_m}^f} = \prod_{i=1}^4{\exp \left( - C_i^{mf} \int_0^{\infty}{\frac{x}{1+\frac{x^{\alpha_N}}{s_f P_m \ell_W C_N d_i^{mf} G\left(x,\theta_{\text{tilt}} \right)} }} dx \right)} \label{LapfmHetNet} \\
&{\mathcal{L}}_{I^f_{\Phi'_f}} = \prod_{i=1}^4{\exp \left(-2\pi \lambda_{f'} p_i^f \int_{\rho}^{\infty}{\frac{x}{1+\frac{x^{\alpha_N}}{s_f P_f \left( \ell_W \right)^2 C_N d_i^f } } dx} \right)}, \label{LapffHetNet}
\end{align}
where $C_i^{mf} = 2\pi \lambda_m p_i^{mf}$ and $g_R \left(\rho \right) = \frac{2\rho}{R_f^2}$. 
\end{theorem}
\begin{proof}
According to the hole point process, the probability that a FBS outside of a sleep region is active equals $\exp \left(-\pi \lambda_m R_c^2\right)$. Thus we have
\begin{equation}
{\mathcal{P}}^c_f \left(\gamma_f , \theta_{\text{tilt}},R_c\right) = \exp\left(-\pi \lambda_m R_c^2 \right) \text{Pr} \left\lbrace \text{SINR}_f > \beta_f \right\rbrace.
\end{equation}
The rest of calculation is similar to theorem \ref{THeo2} and thus is not repeated here.
\end{proof}
In the following lemma, we derive a lower bound on the terms of the coverage probability of the typical femto user.
\begin{lemma}
In the HetNet scenario, ${\mathcal{L}}_{I_{\Phi_m}^f}$ and ${\mathcal{L}}_{I_{\Phi_m^{'}}^f}$ in \eqref{LapfmHetNet} and \eqref{LapffHetNet} can be respectively lower bounded as 
\begin{align}
&{\mathcal{L}}_{I_{\Phi_m}^f} \ge \nonumber \\
&\exp\left( -2\pi \lambda_m \left(s_f P_m \ell_w C_N G_{\text{max}} \right)^{\frac{2}{\alpha_N}} \frac{\pi E\left\lbrace {D^{mf}}^{\frac{2}{\alpha_N}} \right\rbrace}{\alpha_N \sin \frac{2\pi }{\alpha_N}} \right) \label{LapfmHetNetlowbound} \\
&{\mathcal{L}}_{I_{\Phi_m^{'}}^f} \ge \nonumber \\
& \exp\left( -2\pi \lambda_{f'} \left(s_f P_f \left( \ell_w \right)^2 C_N \right)^{\frac{2}{\alpha_N}} \frac{\pi E\left\lbrace {D^f}^{\frac{2}{\alpha_N}} \right\rbrace}{\alpha_N \sin \frac{2\pi}{\alpha_N}} \right). \label{LapffHetNetlowbound}
\end{align}
\end{lemma}
\begin{proof}
To find a lower bound for \eqref{LapfmHetNet}, we replace $G\left(x,\theta \right)$ by its maximum value $G_{\text{max}}$ which results in \eqref{LapfmHetNetlowbound}. Furthermore, to obtain a lower bound for \eqref{LapffHetNet}, we use the fact that for any $f\left(x\right)\ge 0$ and $\rho \ge 0$, $\int_{\rho}^{\infty}{f\left( x \right) dx} \le \int_0^{\infty}{f\left( x \right) dx}$ which results in \eqref{LapffHetNetlowbound}. Also, we have $E\left\lbrace D^{\frac{2}{\alpha_N}} \right\rbrace = \sum_{i=1}^4{p_i d_i^{\frac{2}{\alpha_N}}}$. 
\end{proof}
\begin{corollary}
A lower bound on the coverage probability of the typical femto user in \eqref{FemUsermainCovP} in an interference limited regime (i.e. $\sigma^2 \approx 0$) is given by
\begin{equation}
{\mathcal{P}}_f^c \left(\gamma_f, \theta_{\text{tilt}},R_c \right) \ge C_0 e^{-\pi \lambda_m R_c^2},
\end{equation}
where $C_0$ is obtained as $C_0 = \frac{\alpha_N}{\alpha_L R_f^2\left(C_1+C_2 \right)^{\frac{\alpha_N}{\alpha_L}}}\times $ 
$ \gamma \left(\frac{\alpha_N}{\alpha_L}, \left(C_1 + C_2 \right) R_f^{\frac{2\alpha_L}{\alpha_N}} \right)$, $\gamma\left(. , . \right)$ denotes the lower incomplete gamma function, and $C_1$ and $C_2$ are defined as $C_1 = 2\pi \lambda_m \left( \frac{\gamma_f P_m \ell_w C_N}{P_f C_L D_0^f} \right)^{\frac{2}{\alpha_N}} \frac{\pi}{\alpha_N \sin\frac{2\pi}{\alpha_N}} E\left\lbrace {D^{mf}}^{\frac{2}{\alpha_N}} \right\rbrace$, $C_2 = 2\pi \lambda_{f'} \left( \frac{\gamma_f \left( \ell_w \right)^2 C_N}{C_L D_0^f} \right)^{\frac{2}{\alpha_N}} \frac{\pi}{\alpha_N \sin\frac{2\pi}{\alpha_N}} E\left\lbrace {D^f}^{\frac{2}{\alpha_N}} \right\rbrace$, respectively.
\end{corollary}
\begin{proof}
By substituting \eqref{LapfmHetNetlowbound}, \eqref{LapffHetNetlowbound} into \eqref{FemUsermainCovP} and considering an interference limited regime ($\sigma^2 \approx 0$), the proof is complete. 
\end{proof}

In this scenario, the EE of the network is written as

\begin{equation}\label{EEmacroFemto}
\text{EE}\left( \theta_{\text{tilt}} \right) = \frac{\sum_{i \in \left\{m,f\right\}}{\lambda_i {\mathcal{P}}_i^c\left(\gamma_i, \theta_{\text{tilt}}, R_c \right) \log_2\left(1+\gamma_i \right) }}{\sum_{i \in \left\{m,f\right\}}{\lambda_i\left(P_{ci}+\eta_i P_i \right)}}, 
\end{equation}
where $P_f$ and $P_{cf}$ respectively represent the transmitted power and the constant power consumption in the FBSs and $\eta_f$ is a constant related to the power amplifiers efficiency of the FBSs.

\section{Energy efficiency maximization}\label{EEsec}
As we see in \eqref{EEMBSHom} and \eqref{EEmacroFemto}, the EE is a function of $\theta_{\text{tilt}}$ and therefore, it can be maximized by optimizing the tilt angle. The optimum tilt angle of the BSs is obtained through the following optimization problem
\begin{align}\label{EEMBSmain}
& \underset{\theta_{\text{tilt}}} {\text{maximize}} 
\quad EE(\theta_{tilt}) 
\\ \nonumber
& \quad \text{s.t.} \quad \qquad 0 \le \theta_{\text{tilt}}\le 90^{\circ}. 
\end{align}
\noindent Unfortunately, the objective function of this problem is very complex and in the following, we propose low complexity algorithms for finding the optimal tilt angle in both scenarios.

\subsection{Homogeneous Cellular network}
In fact, to calculate the EE in this scenario, we first need to obtain $E_{I_{\Phi_m}}\left[ e^{-ms I_{\Phi_m}} \sum_{k=0}^{m-1}\frac{\left( ms \left(I_{\Phi_m}+\sigma^2 \right) \right)^k}{k!} \right]$ by \eqref{LaplaceofIntm} for each value of $\rho$. Then, the integral at the numerator \eqref{EEMBSHom2} must be calculated. Hence, the optimum tilt angle can not be found by an efficient method and we have to perform an exhaustive search over all possible values of $\theta_{\text{tilt}}$ which in this case is very hard to implement. 
To address this problem, in the following, we propose a low-complexity method for finding the optimum tilt angle. As it is seen in \eqref{PcMBSmaineq}, for calculating the coverage probability, we need to compute $E_{\rho} {\left\{ e^{-m s \sigma^2} \sum_{k=0}^{m-1}{\sum_{\ell = 0}^k{C_{k,\ell}} \left[\frac{d^{\ell}}{d z^{\ell}} {\mathcal{L}}_I\left(z,\theta_{\text{tilt}} \right) \right]_{z = m s}} \right\}}$. By considering the PDF of $R$ (i.e., the distance between the typical user and its serving BS) in \eqref{PDFNBBL}, we define two distance bounds of $\rho_0$ and $\rho_1$ such that $\text{Pr}\left\{ \rho_0 \le R \le \rho_1 \right\} \ge 1-\epsilon$. Using these bounds, the optimal tilt angle will be restricted to the following range 
\begin{equation} \label{tiltrange}
\max\left\lbrace \text{atan}\left(\frac{H_{\text{eff}}}{\rho_1} \right) - \theta_0 ,0 \right\rbrace \le \theta_{\text{tilt}}\le \text{atan}\left(\frac{H_{\text{eff}}}{\rho_0} \right) + \theta_0 ,
\end{equation}
where $\theta_0 = \theta_{\text{3dB}}\sqrt{{\text{SLL}_{\text{dB}}}/{12}}$. 
The values of $\rho_0$ and $\rho_1$ can be obtained numerically using \eqref{PDFNBBL} for a given $\epsilon$. In Fig. \ref{fig:PDFbound}, we depict the values of these two bounds for $\epsilon = 0.1$ in different densities of the BSs, $\lambda_m$. In addition, the average distance between the typical user and its serving BS, i.e., $\bar{\rho} = E\left\{ \rho \right\}$ is also shown in this figure. It is interesting to note that in large values of $\lambda_m$ (which is related to dense mmWave networks), both of $\rho_0$ and $\rho_1$ converge to $\bar{\rho}$. We exploit this property to simplify the calculations. 
\begin{figure}
\centering
\includegraphics[scale=0.52]{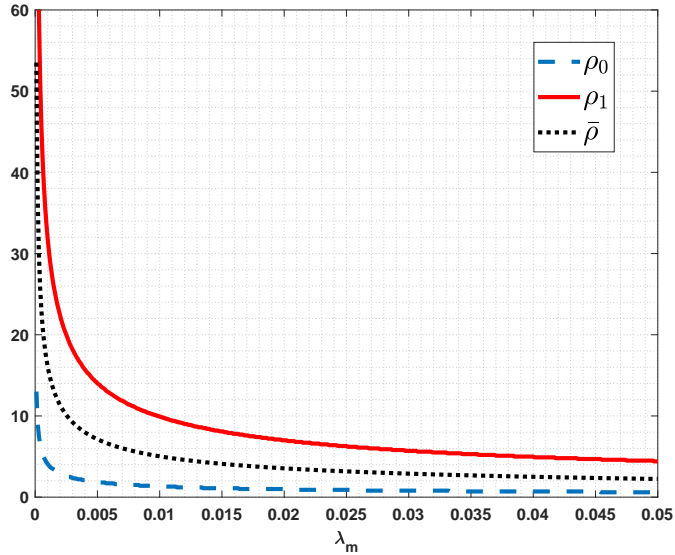}
\caption{~The values of $\rho_0$, $\rho_1$ and $\bar{\rho}$ versus  $\lambda_m$ (for $\beta = 3\times 10^{-3}$) } \label{fig:PDFbound} 
\end{figure}

Using the above property, we can apply the Taylor expansion at the point of $\bar{\rho}$ to obtain the following approximation for \eqref{PcMBSmaineq} as
\begin{align*}
& E_{\rho} {\left\{e^{-m s \sigma^2} \sum_{k=0}^{m-1}{\sum_{\ell = 0}^k{C_{k,\ell}} \left[\frac{d^{\ell}}{d z^{\ell}} {\mathcal{L}}_I\left(z,\theta_{\text{tilt}} \right) \right]_{z = m s}} \right\}} = \nonumber \\
& E_{\rho}\left\{\sum_{n=0}^{\infty}{\frac{\left(\rho - \bar{\rho}\right)^n}{n!}}\times \right. \nonumber \\
& \left. \frac{d^n}{d\rho^n}{\left(e^{-m s \sigma^2} \sum_{k=0}^{m-1}{\sum_{\ell = 0}^k{C_{k,\ell}} \left[\frac{d^{\ell}}{d z^{\ell}} {\mathcal{L}}_I\left(z,,\theta_{\text{tilt}} \right) \right]_{z = m s}} \right)\bigg{\vert}_{\bar{\rho}}}\right\} \approx \nonumber 
\end{align*}
\begin{align}
& e^{-m \bar{s} \sigma^2} \sum_{k=0}^{m-1}{\sum_{\ell = 0}^k{\frac{\left(m \bar{s} \sigma^2 \right)^{k-\ell}}{k!} {k \choose \ell}} \left(-1 \right)^{\ell} \left[\frac{d^{\ell}}{d z^{\ell}} {\mathcal{L}}'_I\left(z,\theta_{\text{tilt}} \right) \right]_{z = m \bar{s}}}, \label{LapappMBS}
\end{align}
\noindent where $\bar{s} = \frac{\gamma \, {\bar{\rho}}^{\, \alpha_L}}{C_L P_o D_0 G_0 }$ and 
\begin{align}\label{ProposApproach}
&{\mathcal{L}}'_{I_{\Phi_m}} \left( z ,\theta_{\text{tilt}} \right) = \nonumber \\
&\prod_{i=1}^4 {\exp \left( -2\pi \lambda_m p_i \int_{\bar{\rho}}^{\infty}{F_L\left(z,x,d_i,\theta_{\text{tilt}} \right) x P_L \left(x \right) dx } \right) } \times \nonumber \\
& \prod_{i=1}^4 {\exp \left( -2\pi \lambda_m p_i \int_{R^{-1}_{eq}(\bar{\rho})}^\infty{F_N\left(z,x,d_i ,\theta_{\text{tilt}}\right) xP_N \left(x \right) dx}\right)}.
\end{align}

In addition to the above approximation, another way to reduce the complexity of the optimization problem in \eqref{EEMBSmain}, is narrowing the search interval of $\theta_{\text{tilt}}$. From \eqref{tiltrange} and considering that in the dense mmWave networks, both 
$\rho_0$ and $\rho_1$ converge to $\bar{\rho}$, we can obtain the bounds of $\theta_{\text{tilt}}$ as $\theta_{\text{min}} \le \theta_{\text{tilt}} \le \theta_{\text{max}}$, where 
\begin{align} \label{Thetminmax}
& \theta_{\text{min}} = \max \left( 0 , \text{atan}\left(\frac{H_{\text{eff}}}{\bar{\rho}} \right) - \theta_0 \right), \nonumber \\ 
& \theta_{\text{max}} = \text{atan}\left(\frac{H_{\text{eff}}}{\bar{\rho}} \right) + \theta_0.
\end{align}
Therefore, an equivalent problem for \eqref{EEMBSmain} can be expressed as
\begin{align}\label{EEMBSmain2}
& \underset{\theta_{\text{tilt}}} {\text{maximize}} \nonumber \\
& \frac{e^{-m \bar{s} \sigma^2} \sum_{k=0}^{m-1}{\sum_{\ell = 0}^k{\frac{\left(m \bar{s} \sigma^2 \right)^{k-\ell}}{k!} {k \choose \ell}} \left(-1 \right)^{\ell} \left[\frac{d^{\ell}}{d z^{\ell}} {\mathcal{L}}'_I\left(z,\theta_{\text{tilt}} \right) \right]_{z = m \bar{s}}} }{P_{c} + \eta_m P_0},
\nonumber \\ 
& \quad \text{s.t.} \quad \qquad \theta_{\text{min}} \le \theta_{\text{tilt}} \le \theta_{\text{max}}. 
\end{align}

This problem has a significantly reduced computational complexity compared to the original problem in \eqref{EEMBSmain}. Since, we do not need to compute \eqref{LaplaceofIntm} for each value of $\rho$. In addition, the search interval is also limited. It can be shown that \eqref{EEMBSmain2} is a convex problem and hence, it can be solved efficiently. In Algorithm \ref{Alg1}, we present a bisection method to solve it. In Section \ref{simsec}, we will show that the performance of the proposed low-complexity approach is very close to the optimal solution found by exhaustive search.
\begin{algorithm} 
\caption{Bisection method}\label{Alg1} 
\begin{algorithmic}[1] 
\STATE Initialize $\theta_{\text{tilt}}^{\text{min}} = \theta_{\text{min}}$ and $\theta_{\text{tilt}}^{\text{max}} = \theta_{\text{max}}$.
%\STATE Calculate $\mathcal{L}_{I_{\Phi_m}}'$ for $\theta_{\text{tilt}} = \text{atan}\left(\frac{H_{\text{eff}}}{\bar{\rho}} \right)$.
\STATE Calculate ${\mathcal{L}}_{I_{\Phi_m}}'$ for $\theta_{\text{tilt}} = \frac{\theta_{\text{tilt}}^{\text{min}} +\theta_{\text{tilt}}^{\text{max}} }{2}$. 
\STATE If resulted ${\mathcal{L}}_{I_{\Phi_m}}$ is greater than the result for $\theta_{\text{tilt}}^{\text{min}}$, then set $\theta_{\text{tilt}}^{\text{min}} =\theta_{\text{tilt}}$. Otherwise set $\theta_{\text{tilt}}^{\text{max}} =\theta_{\text{tilt}}$
\STATE Stop when $\left\lvert \theta_{\text{tilt}}^{\text{min}} - \theta_{\text{tilt}}^{\text{max}} \right\rvert$ is less than a predefined value.
\end{algorithmic} 
\end{algorithm} 
\subsection{Two Tier Heterogeneous Network}
As mentioned in Section \ref{seccov}, in the HetNet scenario, to improve the coverage of the macro users, a sleep region with radius $R_c$ is introduced around each MBS. On the other hand, when we turn off some FBSs, the coverage of the typical femto user decreases. As a result, we have a tradeoff between the coverage probabilities of the macro and femto users. Therefore, in our optimization problem, the radius of the sleep region should be considered as an optimization parameter in addition to the tilt angle, and the EE maximization problem becomes 
\begin{align}\label{mainHetNetEEOpt}
& \underset{\theta_{\text{tilt}},R_c}{\max} \frac{\sum_{i\in \left\{m,f \right\}}{\lambda_i {\mathcal{P}}_i^c \left(\gamma_i , \theta_{\text{tilt}}, R_c\right) \log_2 \left(1+\gamma_i \right)}}{\sum_{i\in \left\{m,f \right\}}{\lambda_i \left( P_{ci}+\eta_i P_i\right)}}, \nonumber \\
& \text{s.t.} \quad {\mathcal{P}}_m^c \ge 1-\epsilon_m, \qquad {\mathcal{P}}_f^c \ge 1-\epsilon_f, \nonumber \\
& 0 \le \theta_{\text{tilt}}\le 90^{\circ}, 0\le R_c \le R_c^{\text{max}},
\end{align}
where $\epsilon_m$ and $\epsilon_f$ are the minimum coverage requirements in the typical macro and femto users, respectively, and $R_c^{\text{max}}$ denotes the maximum radius of the sleep region which is equal to $R_c^{\text{max}} =\frac{1}{\sqrt{\pi \lambda_m}}$. Again, this optimization problem is too complex to solve numerically. 

To reduce the complexity of the above optimization problem, we follow a similar approach as in the homogeneous scenario. To this end, we consider the following optimization problem
\begin{align}
& \underset{\theta_{\text{tilt}},R_c}{\max}\frac{\lambda_m {\mathcal{P}}^{c'}_m \log_2\left(1+\gamma_m \right) +\lambda_{f'} {\mathcal{L}}'_{I^f_{\Phi_m}} {\mathcal{L}}'_{I^f_{\Phi'_f}}\log_2\left(1+\gamma_f \right)}{\lambda_m \left( P_{cm}+\eta_m P_m\right)+\lambda_{f} \left( P_{cf}+\eta_f P_f\right)}, \nonumber \\
& \text{s.t.} \quad {\mathcal{L}}'_{I_{\Phi_m}}{\mathcal{L}}'_{I^m_{\Phi'_f}} \ge 1-\epsilon_m,\: e^{-\pi \lambda_m R_c^2}{\mathcal{L}}'_{I^f_{\Phi_m}} {\mathcal{L}}'_{I^f_{\Phi'_f}} \ge 1-\epsilon_f, \nonumber \\
& 0\le \theta_{\text{tilt}}\le 90^{\circ}, 0\le R_c \le R_c^{\text{max}},
\end{align}
where 
\begin{align*}
{\mathcal{P}}^{c'}_m = \sum_{k=0}^{m-1}\sum_{\ell = 0}^k & \frac{\left(m \bar{s} \sigma^2 \right)^{k-\ell}}{k!}\binom{k}{\ell} \left(-1 \right)^{\ell} \\
&\times \left[\frac{d^{\ell}}{d z^{\ell}} \left( {\mathcal{L}}'_{I_{\Phi_m}}\left(z,\theta_{\text{tilt}} \right) {\mathcal{L}}'_{I^m_{\Phi'_f}}\left(z,\theta_{\text{tilt}} \right) \right) \right]_{z = m\bar{s}}.
\end{align*}
Here we define
\begin{align}
{\mathcal{L}}'_{I^m_{\Phi'_f}} & = \prod_{i=1}^4{\exp\left( -2\pi \lambda_{f'} p_i^{fm}\left( \bar{s} P_f \ell_w C_N d_i^{fm}\right)^{\frac{2}{\alpha_N}} \frac{\pi}{\alpha_N \sin\frac{2\pi}{\alpha_N}} \right)}, \label{LapprimmtoF}\\
{\mathcal{L}}'_{I^f_{\Phi_m}} & =\prod_{i=1}^4{\exp\left( -2\pi \lambda_m p_i^{mf} \int_0^{\infty}{\frac{x dx}{1+\frac{x^{\alpha_N}}{\bar{s}_f P_m \ell_w C_N d_i^{mf} G\left(x,\theta_{\text{tilt}} \right)}}} \right)}, \label{LapprimftoM} \\
{\mathcal{L}}'_{I^f_{\Phi'_f}} & = \prod_{i=1}^4{\exp\left( -2\pi \lambda_{f'} p_i^f \int_{\bar{g}}^{\infty}{\frac{x dx}{1+\frac{x^{\alpha_N}}{\bar{s}_f P_f \left( \ell_w \right)^2 C_N d_i^f }}} \right)}, \label{LapprimftoF}
\end{align}
where $\bar{g} = E\left\lbrace g_R\left( \rho \right) \right\rbrace$ and $\bar{s}_f = \frac{\gamma_f}{P_f C_L D_0} \bar{g}^{\alpha_L}$. By using this approach, the calculations are significantly simplified compared to the optimum exhaustive search. Here, instead of calculating \eqref{MacromainCov} and \eqref{FemUsermainCovP} in which we need to compute \eqref{LaplaceofIntm}, \eqref{LapIntFBStoMacuser}, \eqref{LapfmHetNet} and \eqref{LapffHetNet} for all values of $\rho$, it is sufficient to evaluate \eqref{ProposApproach} and \eqref{LapprimmtoF} only for $\rho = \bar{\rho}$ and obtain \eqref{LapprimftoM} and \eqref{LapprimftoF} only for $\rho = \bar{g}$. In the next section, we will show that by applying this low-complexity approach, only a minor degradation in the performance is observed. 

\section{NUMERICAL RESULT}\label{simsec}
In this section, we numerically evaluate the performance of the proposed 3DBF with tilt angle optimization scheme for the mmWave networks. Through the simulations, we demonstrate that how the EE of the network is improved when the tilt angle of the BSs' antenna is optimized. In addition, the performance of the proposed low-complexity method is compared with the optimal solution obtained by exhaustive search. Table~\ref{Partable} summarizes the simulation parameters used in this section \cite{blockanalysis14,Largantensys15}. 

\begin{table}[!h]
\centering
\caption{~Simulation Parameters } \label{Partable}
\resizebox{\columnwidth}{!}{
\begin{tabular}{ |c|c||c|c|}%{p{0.25\linewidth}p{0.25\linewidth}p{0.25\linewidth}p{0.25\linewidth}}
\hline
\textbf{Parameter} & \textbf{Value} & \textbf{Parameter} & \textbf{Value} \\ 
\hline
SLL$_{dB}$ & $20$ dB & $\theta_{3dB}$ & $6^{\circ}$ \\
\hline
$\alpha_N$ & $4$ & $\alpha_L$ & $2.5$ \\
\hline
$P_f$, $P_{fc} $ & $100$ mWatt, $9.6$ Watt & $\eta_f$ & $4$ \\
\hline
$P_m$, $P_{mc} $ & $20$, $68.73$ Watt & $\eta_m$ & $3.77$ \\
\hline
$\left( M_r,m_r,\theta_r \right)$ & $\left( 10 \text{ dB}, -10 \text{ dB}, 90^{\circ} \right)$ & $\left( M_t,m_t,\theta_t \right)$ 
& $\left( 10 \text{ dB}, -10\text{ dB}, 30^{\circ} \right)$ \\
\hline
$\left( M_r^f,m_r^f,\theta_r^f \right)$ & $\left( 10 \text{ dB}, -10 \text{ dB}, 90^{\circ} \right)$ & $\left( M_t^f,m_t^f,\theta_t^f \right)$ & $\left( 10 \text{ dB}, -10\text{ dB}, 30^{\circ} \right)$ \\
\hline
$\beta_1, \beta_2$ & $0.003, 0.006$ & $R_f$ & $ 30 \text{ m}$ \\
\hline
\end{tabular}
}
\end{table}

We first examine the homogeneous scenario. The coverage probability of the network under this scenario is depicted in Fig.~\ref{fig:ProposAppCov} as a function of the SINR threshold. The density of the MBS is $\lambda_m = 4.973 \times 10^{-5}$ and the curves are obtained under two different values of the blockage effect intensity $\beta$ as in table~\ref{Partable} and $m=5$. It is observed that by increasing the 3DBF outperforms in comparison with the network in which the tilt angle is not optimized (marked as 2DBF in the figure) and also the proposed low complexity method have performance close the optimal solution resulted from the exhaustive search.
\begin{figure}[!t]

\centering 
\includegraphics[scale=0.54]{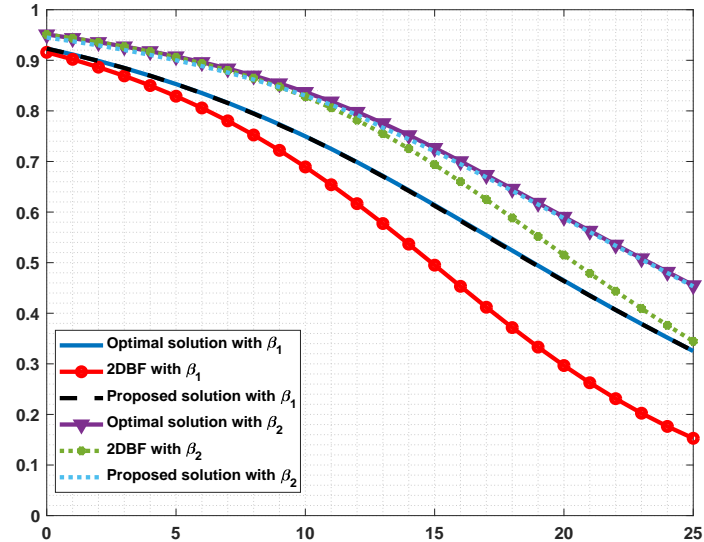}
\caption{~ Comparison of the coverage probability of the proposed low-complexity approach and the optimal solution and with the 2DBF.} \label{fig:ProposAppCov}
\end{figure}

Fig.~\ref{fig:ProposApp} illustrates the EE of the network in terms of the SINR threshold for $\lambda_m = 8\times 10^{-4}$ and two values of $\beta$ as in Table~\ref{Partable} under this scenario with $m=1$. As we see in this figure, the EE of the network that adopts 3DBF is always improved in comparison with the 2DBF. This improvement is more than 100\% in high SINR thresholds. In addition, in this figure, the EE performance of the proposed low-complexity method is compared with the optimum method based on the exhaustive search. As we see in the figure, the performance of the proposed low-complexity approach is the same as the optimal solution in almost all the SINR threshold and for both values of $\beta$. 
\begin{figure}[!h]
\centering 
\includegraphics[scale=0.54]{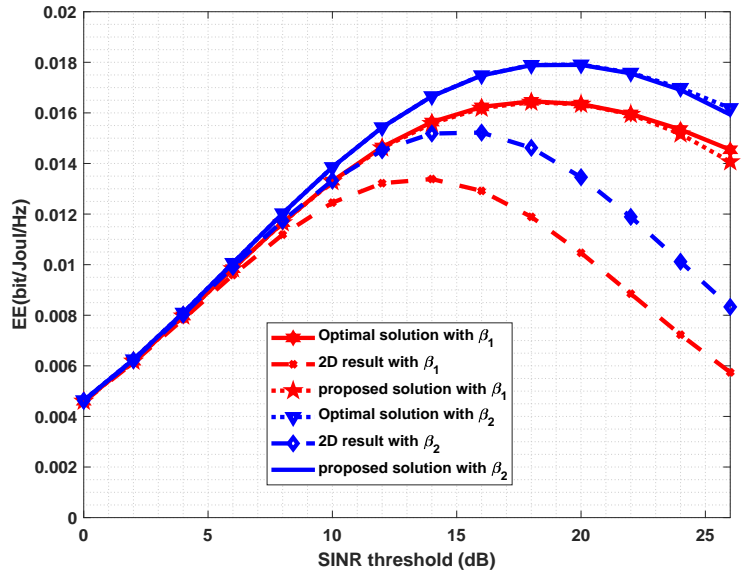}
\caption{~EE comparison of the proposed low-complexity approach and the optimal solution and with the 2DBF.} \label{fig:ProposApp}
\end{figure}

Fig.~\ref{fig:EElambda2_1} presents the network EE with respect to the tilt angle for $\lambda_m = 5.093\times 10^{-6}$ and $m=5$. In this figure, the optimum tilt angles obtained by exhaustive search and the proposed low-complexity method are shown. Also, the dashed lines represent the tilt angle bounds obtained in \eqref{Thetminmax}. We see that both tilt angles are almost the same.

\begin{figure}[t]
\centering
\includegraphics[scale = 0.54]{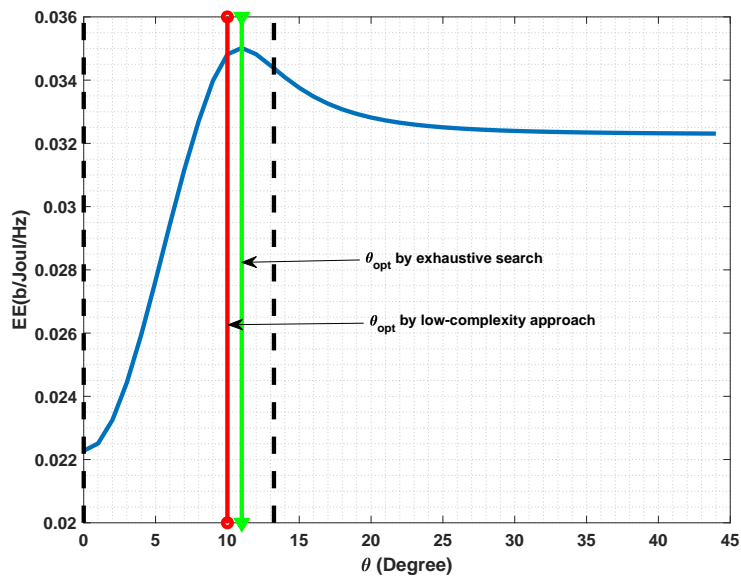}
\caption{~EE comparison with respect to the BS tilt angle $\beta=\beta_1$, $\gamma = 20 \text{ dB}$, $\lambda_m = 5.093 \times 10^{-6}$.} \label{fig:EElambda2_1} 
\end{figure}

In Fig.~\ref{fig:MacroCovbei1Lamb1}, performance of the HetNet scenario is evaluated. This figure exhibits the effect of the FBSs density $\lambda_f$ and the radius of the sleep region $R_c$ on the optimum tilt angle that maximizes the coverage of the typical macro user. In this figure we see that by increasing the density of the FBSs, the optimum tilt angle slightly decreases. Also by increasing $R_c$ or reducing the density of the FBSs, the coverage probability of the macro users increases, since interference from the FBSs is reduced.

\begin{figure}[t]
\centering 
\includegraphics[scale=0.54]{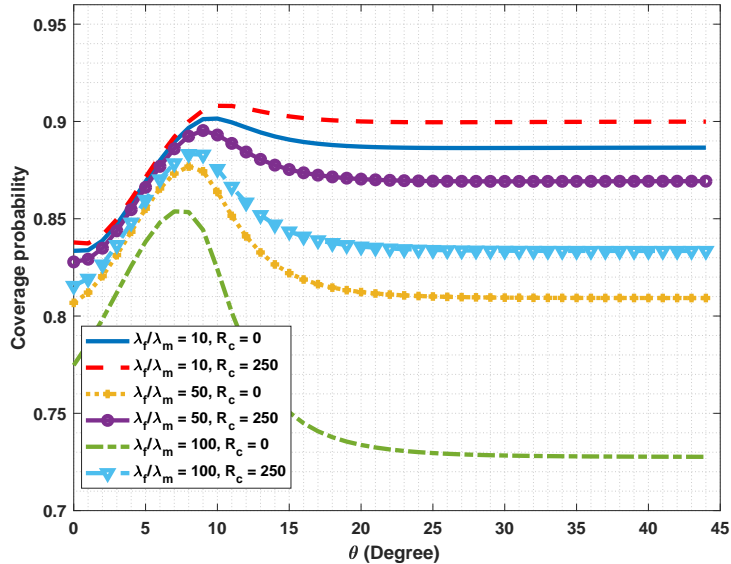}
\caption{~The coverage probability of the typical macro user with various $\lambda_f$ and $R_c$ ($\lambda_m = 5.093 \times 10^{-6}, \beta = \beta_1$).} \label{fig:MacroCovbei1Lamb1}
\end{figure}

Fig.~\ref{fig:Lowersig5e-4} illustrates the coverage probability of the typical femto user in terms of $R_c$ for case of $\sigma^2 = 5 \times 10^{-4}$, which corresponds to $\text{SNR}_f= \frac{P_f}{\sigma^2} = 23\text{ dB}$. In this figure, we see that the lower bound obtained in (30) is very tight. It is observed that although the lower bound is obtained under assumption of an interference limited scenario, it is still quite tight in other scenarios.
\begin{figure}[t]
\centering 
\includegraphics[scale=0.54]{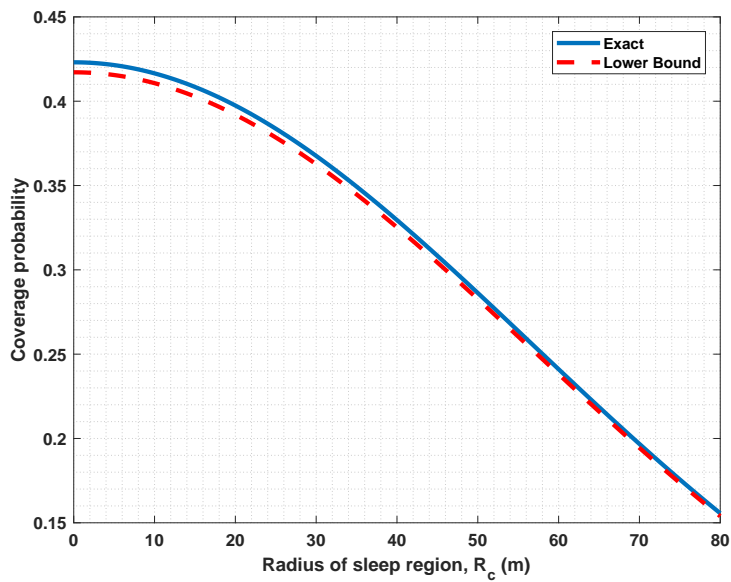}
\caption{~Comparison of the exact coverage probability of the typical femto user and its lower bound with $\lambda_f = 10 \lambda_m (\lambda_m = 4.973 \times 10^{-5})$, $\sigma^2 = 5 \times 10^{-4} (\text{SNR}_f = 23\text{ dB})$.} \label{fig:Lowersig5e-4}
\end{figure}
Figures \ref{fig:OptEEi3bei1bei2} and \ref{fig:Rcopti3bei1bei2} plot the optimal EE, the optimal radius of the sleep region and the optimal tilt angle, respectively, with $\epsilon_m = 0.2, \epsilon = 0.7$, $\gamma_m = \gamma_f = 10\text{ dB}$. We can check that the proposed low complexity approach has only a minor degradation in the performance with respect to the exhaustive search.
\begin{figure}[!h]
%\begin{centering}
\centering 
\includegraphics[scale=0.54]{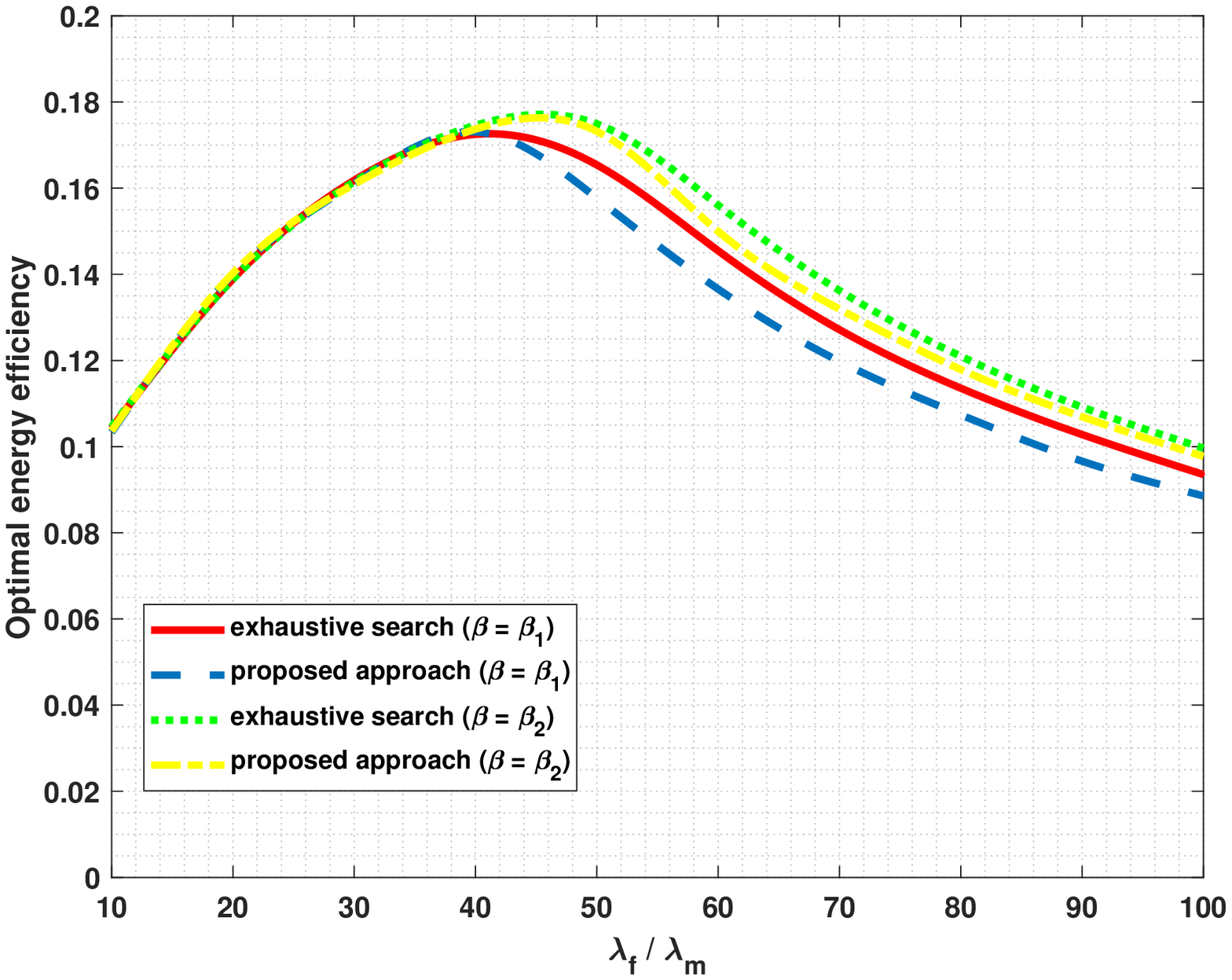}
\caption{~Comparison of the optimal EE for different values of the blockage parameter ($\lambda_m =  4.973 \times 10^{-5}$) .} \label{fig:OptEEi3bei1bei2}
%\end{centering}
\end{figure}

\begin{figure}[!h]
%\begin{centering}
\centering 
\includegraphics[scale=0.54]{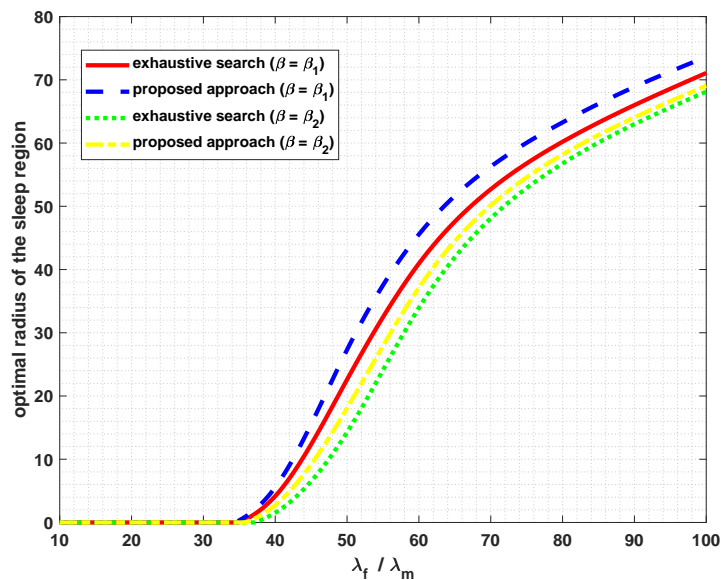}
\caption{~Comparison of the optimal $R_c$ for different values of the blockage parameter ($\lambda_m =  4.973 \times 10^{-5}$) .} \label{fig:Rcopti3bei1bei2}
%\end{centering}
\end{figure}

\section{CONCLUSION}\label{consec}
In this paper, we have studied the EE maximization problem in the downlink of a 3D beamforming mmWave network. We have optimized the tilt angle of the BSs to maximize the EE in a homogeneous network and a two tier HetNet mmWave cellular network. In both scenarios, we have optimized the tilt angle of the MBS’s antenna to maximize the EE. In addition, in the second scenario, the optimization of the radius of the sleep region has also been considered. In addition, to reduce the complexity of the optimization problems, an efficient method based on bisection algorithm has been proposed to compute the optimal tilt angle. The proposed algorithms result in almost the same EE performance as the optimal method based on exhaustive search but with much reduced complexity.

\bibliographystyle{IEEE}

\begin{thebibliography}{99}

\bibitem{itwillwork13}
T.~Rappaport, S.~Sun, R.~Mayzus, H.~Zhao, Y.~Azar, K.~Wang, G.~Wong, J.~Schulz,
M.~Samimi, and F.~Gutierrez, ``Millimeter wave mobile communications for 5G
cellular: It will work!'' \emph{IEEE Access}, vol.~1, pp. 335--349, May. 2013.

\bibitem{covratemmwave15}
T.~Bai and R.~W. Heath, ``Coverage and rate analysis for millimeter-wave
cellular networks,'' \emph{IEEE Trans. Wireless Commun.},
vol.~14, no.~2, pp. 1100--1114, Feb. 2015.

\bibitem{blockanalysis14}
T.~Bai, R.~Vaze, and R.~W. Heath, ``Analysis of blockage effects on urban
cellular networks,'' \emph{IEEE Trans. Wireless Commun.},
vol.~13, no.~9, pp. 5070--5083, Sept. 2014.

\bibitem{mmWaveEn14}
W.~Roh, J.~Seol, J.~Park, B.~Lee, J.~Lee, Y.~Kim, J.~Cho, K.~Cheun, and
F.~Aryanfar, ``Millimeter-wave beamforming as an enabling technology for 5G
cellular communications: theoretical feasibility and prototype results,''
\emph{IEEE Commun. Mag.}, vol.~52, no.~2, pp. 106--113, Feb. 2014.

\bibitem{JCNmmWave16}
Y.~Niu, C.~Gao, Y.~Li, L.~Su, and D.~Jin, ``Exploiting multi-hop relaying to
overcome blockage in directional mmWave small cells,'' \emph{Journal of
Commun. and Net. (JCN)}, vol.~18, no.~3, pp. 364--374, Jun. 2016.

\bibitem{DrRazaviTr}
S.~M. Razavizadeh, M.~Ahn, and I.~Lee, ``Three-dimensional beamforming: A new
enabling technology for 5G wireless networks,'' \emph{IEEE Signal Process. Mag.}, vol.~31, no.~6, pp. 94--101, Nov. 2014.

\bibitem{3DBFdesignLee13}
W.~Lee, S.~R. Lee, H.-B. Kong, and I.~Lee, ``3D beamforming designs for single
user MISO systems,'' in \emph{Proc. IEEE Global Commun. Conf.
(GLOBECOM)}, pp. 3914--3919, Dec. 2013.

\bibitem{coveragemagheath14}
T.~Bai, A.~Alkhateeb, and R.~Heath, ``Coverage and capacity of millimeter-wave
cellular networks,'' \emph{IEEE Commun. Mag.}, vol.~52, no.~9,
pp. 70--77, Sept. 2014.

\bibitem{SGmodelingElsawy17}
H.~Elsawy, A.~Sultan-Salem, M.~S. Alouini, and M.~Z. Win, ``Modeling and
analysis of cellular networks using stochastic geometry: A tutorial,''
\emph{IEEE Commun. Surveys Tuts.}, vol.~19, no.~1, pp. 167--203,
First quarter 2017.

\bibitem{JCNLTEA18}
G.~Giambene and V.~A. Le, ``Analysis of LTE-A heterogeneous networks with
SIR-based cell association and stochastic geometry,'' \emph{Journal of
Commun. and Net. (JCN)}, vol.~20, no.~2, pp. 129--143, Apr. 2018.

\bibitem{SGmodmultitier15}
M.~D. Renzo, ``Stochastic geometry modeling and analysis of multi-tier
millimeter wave cellular networks,'' \emph{IEEE Trans. Wireless Commun.}, vol.~14, no.~9, pp. 5038--5057, Sept. 2015.

\bibitem{CovHetnetDown17}
E.~Turgut and M.~C. Gursoy, ``Coverage in heterogeneous downlink millimeter
wave cellular networks,'' \emph{IEEE Trans. Commun.},vol.~65, no.~10, pp. 4463--4477, Oct. 2017.

\bibitem{Onireti2018}
O.~Onireti, A.~Imran, and M.~A. Imran, ``Coverage, capacity, and energy
efficiency analysis in the uplink of {mmWave} cellular networks,''
\emph{{IEEE} Trans. Veh. Technol.}, vol.~67, no.~5, pp.
3982--3997, May 2018.

\bibitem{Cen2017}
S.~Cen, X.~Zhang, M.~Lei, S.~Fowler, and X.~Dong, ``Stochastic geometry
modeling and energy efficiency analysis of millimeter wave cellular
networks,'' \emph{Wireless Networks}, vol.~24, no.~7, pp. 2565--2578, Mar.
2017.

\bibitem{3GPPstd}
\emph{3GPP TR 36.814 V9.0}, ``Further advancements for E-UTRA physical layer aspects,'' Tech. Rep. 2010.

\bibitem{tiltangtwo15}
R.~Hernandez~Aquino, S.~Zaidi, D.~McLernon, M.~Ghogho, and A.~Imran, ``Tilt
angle optimization in two-tier cellular networks - a stochastic geometry
approach,'' \emph{IEEE Trans. Commun.}, vol.~63, no.~12, pp.
5162--5177, Dec. 2015.

\bibitem{SGwireles09}
%\BIBentryALTinterwordspacing
F.~Baccelli and B.~Blaszczyszyn, \emph{{Stochastic Geometry and Wireless
Networks, Volume I and II}}, Foundations and Trends in Networking: Vol.
4: No 1-2, pp 1-312.\hskip 1em plus 0.5em minus 0.4em\relax {NoW Publishers}, vol.~2. [Online]. Available: \url{https://hal.inria.fr/inria-00403040},
%\BIBentrySTDinterwordspacing

\bibitem{rtackappandrew11}
J.~Andrews, F.~Baccelli, and R.~Ganti, ``A tractable approach to coverage and
rate in cellular networks,'' \emph{IEEE Trans. Commun.},
vol.~59, no.~11, pp. 3122--3134, Nov. 2011.

\bibitem{Ktierhetnet12}
H.~S. Dhillon, R.~K. Ganti, F.~Baccelli, and J.~G. Andrews, ``Modeling and
analysis of K-tier downlink heterogeneous cellular networks,'' \emph{IEEE J. Sel. Areas Commun.}, vol.~30, no.~3, pp. 550--560,
Apr. 2012.

\bibitem{HaenggiLarge09}
M.~Haenggi and R.~K. Ganti, ``Interference in large wireless networks,''
\emph{Found. Trends Netw.}, vol.~3, no.~2, pp. 127--248, Feb. 2009.

\bibitem{EETotuGen}
J.~Wu, Y.~Zhang, M.~Zukerman, and E.~K.~N. Yung, ``Energy-efficient
base-stations sleep-mode techniques in green cellular networks: A survey,''
\emph{IEEE Commun. Surveys Tuts.}, vol.~17, no.~2, pp. 803--826,
Second quarter 2015.

\bibitem{Largantensys15}
S.~Han, C.~I, Z.~Xu, and C.~Rowell, ``Large-scale antenna systems with hybrid
analog and digital beamforming for millimeter wave 5G,'' \emph{IEEE
Commun. Mag.}, vol.~53, no.~1, pp. 186--194, Jan. 2015.
\end{thebibliography}

\end{document}